\documentclass[conference]{IEEEtran}
\pdfoutput=1

\usepackage[sorting=none,style=ieee,backend=bibtex]{biblatex}
\usepackage{amsmath}
\usepackage{amsbsy}
\usepackage[noend]{algpseudocode}
\usepackage{graphicx}
\usepackage{mathtools}
\usepackage{amssymb}
\usepackage{eqnarray}
\usepackage{multicol}
\usepackage{subcaption}
\usepackage{listings}
\usepackage{color}
\usepackage{units}
\usepackage{amsthm}
\usepackage{amsmath}
\usepackage{vwcol}
\usepackage{caption}
\usepackage{footnote}
\usepackage{blkarray}
\usepackage{comment}
\usepackage{tikz}
\usepackage{textcomp}
\usepackage{hyperref}

\newtheorem{theorem}{Theorem}

\newtheorem{lemma}{Lemma}
\newtheorem{example}{Example}[section]
\newtheorem{assumption}{Assumption}
\def\tablename{Table}
\makesavenoteenv{tabular}
\makesavenoteenv{table}

\allowdisplaybreaks

\ifCLASSINFOpdf

\newcommand\copyrighttext{%
  \footnotesize \textcopyright IEEE. Personal use of this material is permitted. However, permission to reprint/republish this material for advertising or promotional purposes or for creating new collective works for resale or redistribution to servers or lists, or to reuse any copyrighted component of this work in other works must be obtained from the IEEE.}
\newcommand\copyrightnotice{%
\begin{tikzpicture}[remember picture,overlay]
\node[anchor=south,yshift=10pt] at (current page.south) {\fbox{\parbox{\dimexpr\textwidth-\fboxsep-\fboxrule\relax}{\copyrighttext}}};
\end{tikzpicture}%
}

\newcommand{\mc}[1]{\ensuremath{\mathcal{#1}}}

\newcommand{\cyx}{\ensuremath{C_{Y\rightarrow X}}}

\newcommand{\estcyx}{\ensuremath{\hat{C}_{Y\rightarrow X}}}

\newcommand{\optcyx}{\ensuremath{C^*_{Y\rightarrow X}}}
\newcommand{\fxr}[1]{\ensuremath{f_{X_#1}^{(r)}}}
\newcommand{\fxc}[1]{\ensuremath{f_{X_#1}^{(c)}}}

\newcommand{\estfxr}[1]{\ensuremath{\hat{f}_{X_#1}^{(r)}}}
\newcommand{\estfxc}[1]{\ensuremath{\hat{f}_{X_#1}^{(c)}}}

\newcommand{\optfxr}[1]{\ensuremath{f_{X_#1}^{(r)*}}}
\newcommand{\optfxc}[1]{\ensuremath{f_{X_#1}^{(c)*}}}

\newcommand{\history}[1]{\ensuremath{\mc{H}_#1}}
\newcommand{\hr}[1]{\ensuremath{\history{i}^{(r)}}}
\newcommand{\hc}[1]{\ensuremath{\history{i}^{(c)}}}
\newcommand{\kl}[2]{\ensuremath{D(#1 \mid \mid #2)}}

\newcommand{\refclass}{\ensuremath{\tilde{\mc{P}}}}
\newcommand{\refclassr}{\ensuremath{\refclass^{(r)}}}
\newcommand{\refclassc}{\ensuremath{\refclass^{(c)}}}
\newcommand{\mcal}[1]{\ensuremath{\mathcal{#1}}}
\begin{document}

\title{A Sample Path Measure of Causal Influence}

\author{
Gabriel~Schamberg, \emph{Student Member, IEEE},
Todd~P.~Coleman, \emph{Senior Member,~IEEE}
}

\maketitle

\copyrightnotice

%%%%%%%%%%%%%%%%%%%%%%%%%%%%%%%%%%%%%%%%%%%%%%%%%%%%%%%%%%%%%%%%%%%%%%%%%%%%%%%

\begin{abstract}
We present a sample path dependent measure of causal influence between two time series. The proposed measure is a random variable whose expected sum is the directed information. A realization of the proposed measure may be used to identify the specific patterns in the data that yield a greater flow of information from one process to another, even in stationary processes. We demonstrate how sequential prediction theory may be leveraged to obtain accurate estimates of the causal measure at each point in time and introduce a notion of regret for assessing the performance of estimators of the measure. We prove a finite sample bound on this regret that is determined by the regret of the sequential predictors used in obtaining the estimate. We estimate the causal measure for a simulated collection of binary Markov processes using a Bayesian updating approach. Finally, given that the measure is a function of time, we demonstrate how estimators of the causal measure may be extended to effectively capture causality in time-varying scenarios.
\end{abstract}

%%%%%%%%%%%%%%%%%%%%%%%%%%%%%%%%%%%%%%%%%%%%%%%%%%%%%%%%%%%%%%%%%%%%%%%%%%%%%%%


\IEEEpeerreviewmaketitle

%%%%%%%%%%%%%%%%%%%%%%%%%%%%%%%%%%%%%%%%%%%%%%%%%%%%%%%%%%%%%%%%%%%%%%%%%%%%%%%

\section{Introduction}

\begin{comment}
- Granger Causality
- Directed Information
- Transfer Entropy
- Local Transfer Entropy
- Sequential Prediction
- Argument for why sample path causality is important. Perhaps a neuroscientific argument that even if the model is time invariant, there may be certain patterns that induce a greater level of causal influence.
\end{itemize}}
\end{comment}

In 1969, Granger \cite{granger1969investigating} built upon the the ideas of Wiener by proposing an approach to identify causal relationships between time series. While his original treatment was applied only to linear regression models, his underlying perspective that a time series $Y^n$ is ``causing'' $X^n$ if we can better predict $X^n$ given all information than given all information excluding $Y^n$ is still utilized throughout causality research. More modern information theoretic interpretations of this principle include directed information (DI) \cite{marko1973bidirectional,massey1990causality} and transfer entropy (TE) \cite{schreiber2000measuring}, which is equivalent to Granger causality (GC) for Gaussian autoregressive processes \cite{barnett2009granger}. Both of these quantities measure the reduction in uncertainty (i.e. conditional entropy) of the future of $X^n$ that is obtained by including the past of $Y^n$ in the available information in an appropriate sense. Interestingly, both quantities are determined by taking expectations over all sequences, and thus are dependent solely on a system's underlying distribution and not a given realization of the collection of processes.

These quantities may be adjusted to incorporate a notion of locality through use of \emph{self-information}. For a given realization $x$ of a random variable $X\sim f_X$, the self-information is given by $h(x) \triangleq -\log f_X(x)$ and represents the amount of surprise associated with that realization. By replacing entropy with self-information, and its conditional form $h(x\mid y) \triangleq -\log f_{X\mid Y}(x\mid y) $, local versions of DI, TE, and their conditional extensions may be obtained (see Table 1 in \cite{lizier2014jidt}). While the local extensions of DI and TE are indeed dependent on realizations, they may take on negative values. Such a scenario occurs when the knowledge that $Y=y$ makes the observation of $X=x$ less likely to have occurred, i.e. $f_X(x) > f_{X\mid Y}(x\mid y)$. While an interesting concept, it is not clear how to interpret negative values in the context of assessing the presence/absence of a causal link.

As such, estimating the causal structure (i.e. directed graph) of a collection of processes typically involves estimating an averaged measure such as DI \cite{quinn2015directed,amblard2011directed}. In a time varying scenario, however, it would be necessary to replace an expectation over time with some sort of windowing technique as in \cite{oselio2017dynamic}. As a result, estimates of this form may be able to capture changes in the underlying system model, but do not reflect the varying levels of causal influence that occur within windows of time for which there is stationarity (see Example \ref{example}). Here we build on the causal inference perspective presented in \cite{kim2014dynamic} and propose a causal measure that captures changes in time without requiring a windowing approach. Furthermore, we introduce a framework for estimating the causal measure using sequential prediction and derive a finite sample bound on the accuracy of such estimators.
\section{Sample Path Measure of Causal Influence} \label{measure}
\begin{comment}
- Definition
- True causal measure is random
- Can be applied to time varying settings
- Always non-negative
- Is always well defined (as opposed to directed information rate)
- ``Semi-local'' - Measure such as DI and TE are functions of the model (can't capture local influence), but local measures like pointwise mutual information can be negative. Our measure is local and positive.
\end{comment}

Suppose we observe the stochastic processes $X^n \in \mcal{X}^n$, $Y^n \in \mcal{Y}^n$, and $Z^n \in \mcal{Z}^n$, characterized by the joint probability mass function (pmf) $f_{X^n,Y^n,Z^n}(x^n,y^n,z^n)$. Although this work applies more generally, for the purpose of exposition we only consider discrete probability measures. We begin by considering the scenario where, having observed $(x^{i-1},y^{i-1},z^{i-1})$, we wish to determine the causal influence that $y^{i-1}$ has on the next observation $x_i$. In such a scenario, we consider the following \emph{restricted} (denoted $(r)$) and \emph{complete} (denoted $(c)$) conditional distributions:
\begin{align}
\fxr{i}(x_i) &\triangleq f_{X_i \mid X^{i-1},Z^{i-1}}
    (x_i \mid x^{i-1},z^{i-1}) \\
\fxc{i}(x_i) &\triangleq f_{X_i \mid X^{i-1},Y^{i-1},Z^{i-1}}
    (x_i \mid x^{i-1},y^{i-1},z^{i-1}).
\end{align}

\noindent Using these distributions, at each time $i$ we define the sample path measure of causality from $Y$ to $X$ in the presence of side information $Z$ for given realizations $(x^{i-1},y^{i-1},z^{i-1})$ as:
\begin{equation}
\cyx(x^{i-1},y^{i-1},z^{i-1}) = \kl{\fxc{i}}{\fxr{i}}.
\end{equation}

\noindent For ease of notation, we may represent the causal measure at time $i$ simply as $\cyx(i)$.

The key observation that must be made is that \fxc{i} and \fxr{i} are determined by the realizations of $X^n$, $Y^n$, and $Z^n$. As a result, \emph{the causal measure is a random variable}. In this regard, our causal measure is different from previous measures of causality wherein the causal influence is determined by the model, and not the sample path. To ensure this point is made clear, we will present an example.

\begin{example}\label{example}
Suppose $Y_i \sim \text{Bern}(0.2)$ iid for $i=1,2,\dots$ and:
\begin{equation*}
X_i \sim
\begin{cases}
      \text{Bern}(0.9), & Y_{i-1} = 1 \\
      \text{Bern}(0.5), & Y_{i-1} = 0
\end{cases}
\end{equation*}

\noindent Intuitively, we would expect that in some sense $Y^n$ is ``causing'' $X^n$ to a greater extent when $Y_i$ is one than when it is zero. In order to formalize this, we have to find the probability of $X_i=1$ when only $X^{i-1}$ is known (i.e. the restricted distribution):
\begin{equation*}
\begin{aligned}
\mathbb{P}(&X_i = 1 | X^{i-1}=x^{i-1}) \\
&= \mathbb{P}(X_i = 1) \\
&= \sum_{y_{i-1}\in \{0,1\}}
    \mathbb{P}(X_i =1 \mid Y_{i-1} = y_{i-1}) \mathbb{P}(Y_{i-1} = y_{i-1}) \\
&= (0.5)(0.8) + (0.9)(0.2) \\
&= 0.58.
\end{aligned}
\end{equation*}

\noindent We can fully characterize the complete and restricted probability mass functions (pmfs) using these probabilities, i.e. $\fxr{i}(1) = 0.58$, $\fxc{i}(1) = 0.9$ if $y_{n-1}=1$, and $\fxc{i}(1) = 0.5$ if $y_{n-1}=0$. We can now compute the causal measure, which takes on one of two values determined by the observation $y_{i-1}$:
\begin{equation*}
\cyx(i) =
\begin{cases}
      0.363, & y_{i-1} = 1 \\
      0.019, & y_{i-1} = 0
\end{cases}
\end{equation*}

\noindent Thus, we see that our measure captures how, even in a stationary Markov chain, different patterns in the observed data may give rise to different levels of causal influence. By contrast, we note that because the process is stationary, the directed information rate and transfer entropy are both given simply by $E[\cyx]=(0.9)(0.019)+(0.1)(0.363)=0.088$.
\end{example}

The above example gives rise to two key observations. First, even stationary Markov processes exhibit dynamic causal behaviors that are not captured when taking an outer expectation. Second, by averaging over all possible histories, TE and DI are minimally affected by patterns that occur with low probability, even if those patterns induce a high level of causal influence.

We now discuss some key properties of the proposed causal measure. First, we note the crucially important quality of non-negativity, which follows directly from the non-negativity of KL-Divergence. Next, we characterize our measure as being ``semi-local.'' We note that GC, DI (rate), TE, etc. are all \emph{expectations}, determined entirely by the underlying probabilistic model of the observed data. In this regard, these measures are not local representations of the observed data. On the other end of the spectrum, local data-dependent versions of these measures may be obtained by substituting the self information for entropy, but these local versions of causal measures may be negative when unlikely sequences occur. Our measure is ``semi-local'' in the sense that at any given time, the measure is determined by the observations from the past, but guarantees non-negativity by taking an expectation over the future.

\begin{comment}
\noindent This process can be equivalently characterized by the four-state ``complete'' Markov Chain with states $(X,Y) \in \{0,1\}^2$ and transition matrix:

\begin{equation}
M^{(c)} =
\begin{blockarray}{ccccc}
(0,0) & (0,1) & (1,0) & (1,1) \\
\begin{block}{(cccc)c}
    0.25 & 0.1 & 0.25 & 0.1 & (0,0) \\
    0.25 & 0.9 & 0.25 & 0.9 & (0,1) \\
    0.25 & 0.1 & 0.25 & 0.1 & (1,0) \\
    0.25 & 0.9 & 0.25 & 0.9 & (1,1) \\
\end{block}
\end{blockarray}
\end{equation}

\noindent where $M^{(c)}_{kj}$ represents the probability of having $(X_i,Y_i)$ in state $k$ given that $(X_{i-1},Y_{i-1})$ is in state $j$. We can additionally define a two-state ``restricted'' Markov Chain with states $X \in \{0,1\}$ and transition matrix:

\begin{equation}
M^{(r)} =
\begin{blockarray}{ccc}
0 & 1 \\
\begin{block}{(cc)c}
    0.3 & 0.3 & 0 \\
    0.7 & 0.7 & 1 \\
\end{block}
\end{blockarray}
\end{equation}
\end{comment}
\section{Estimation of the Causal Measure} \label{estimation}

\begin{comment}
- Sequential Prediction
- Notion of causality regret - whereas DI has one value that the Tsachy paper tries to converge to, we are looking more at a sequential prediction problem because we want to estimate causality accurately at every point in time
- Theorem
- Independent selection of restricted and complete reference classes - addresses problem with restricted distribution being infinite order raised in Purdon paper
- Discussion of estimating with time varying statistics
- Can be computed online
\end{comment}

An estimate of the causal measure can be obtained by simply estimating the complete and restricted distributions and then computing the KL divergence between the two at each time. Such an estimator allows us to leverage results from the field of sequential prediction. The sequential prediction problem formulation we consider is as follows: for each round $i \in \{1,\dots,n\}$, having observed some history $\history{i}$, a learner selects a probability assignment $\hat{f}_i \in \mc{P}$, where $\mc{P}$ is the space of probability distributions over $\mc{X}$. Once $\hat{f}_i$ is chosen, $x_i$ is revealed and a loss $l(\hat{f}_i,x_i)$ is incurred by the learner, where the loss function $l:\mc{X}\rightarrow \mathbb{R}$ is chosen to be the self-information loss given by $l(f,x) = -\log f(x)$.

The performance of sequential predictors may be assessed using a notion of \emph{regret} with respect to a reference class of probability distributions $\refclass \subset \mc{P}$. For a given round $i$ and reference distribution $\tilde{f}_i \in \refclass$, the learner's regret is:

\begin{equation}
r(\hat{f}_i,\tilde{f}_i,x_i) = l(\hat{f}_i,x_i) - l(\tilde{f}_i,x_i)
\end{equation}

\noindent In many cases the performance of sequential predictors will be measured by the worst case regret, given by:

\begin{align}
R_n(\refclass_n) &= \sup_{x^n \in \mc{X}^n} \sum_{i=1}^n l(\hat{f}_i,x_i) - \inf_{\tilde{f}\in \refclass_n} \sum_{i=1}^n l(\tilde{f}_i,x_i) \label{optimal_f} \\
&\triangleq \sup_{x^n \in \mc{X}^n} \sum_{i=1}^n r(\hat{f}_i,f^*_i,x_i)
\end{align}

\noindent where $f^*_i \in \refclass$ is defined as the distribution from the reference class with the smallest cumulative loss up to time $n$, i.e. the $\tilde{f}_i$ for which $R_n$ is largest. We also define $f^* \in \refclass_n \subset \mc{P}^n$ to be the cumulative loss minimizing \emph{joint} distribution, noting that the reference class of joint distributions $\refclass_n$ is not necessarily equal to $\refclass^n$ (i.e. $\refclass \times \refclass \times \dots$), as often times there may be a constraint on the selection of the best reference distribution that is imposed in order to establish bounds. In the absence of any restrictions, the reference distributions may be selected at each time such that $f^*_i(x_i)=1$, resulting in zero cumulative loss for any sequence $x^n$. Thus, bounds on regret often assume stationarity by enforcing $f_1^*=f_2^*=\dots=f_n^*$ or assume that $f_i^* = f^*_{i+1}$ for all but some small number of indices. For various learning algorithms (i.e. strategies for selecting $\hat{f}_i$ given $\history{i}$) and reference classes $\refclass_n$, these bounds on the worst case regret are defined as a function of the sequence length $n$:

\begin{equation}
R_n(\refclass_n) \le M(n)
\end{equation}

It follows naturally that an estimator for our causal measure can be constructed by building two sequential predictors. The restricted predictor $\estfxr{i}$ computed at each round using $\hr{i} \triangleq \{x_1,\dots,x_{i-1}\} \cup \{z_1,\dots,z_{i-1}\}$, and the complete predictor $\estfxc{i}$ computed at each round using $\hc{i} \triangleq \{x_1,\dots,x_{i-1}\} \cup \{y_1,\dots,y_{i-1}\} \cup \{z_1,\dots,z_{i-1}\}$. It then follows that each of these predictors will have an associated worst case regret, given by $R^{(r)}_n(\refclassr_n)$ and $R^{(c)}_n(\refclassc_n)$, where $\refclassr_n$ and $\refclassc_n$ represent the restricted and complete reference classes. Using these sequential predictors, we define our estimated causal influence from $Y$ to $X$ at time $i$ as:
\begin{equation}
\estcyx(i) = \kl{\estfxc{i}}{\estfxr{i}}
\end{equation}

\noindent It should be noted that when averaged over time, this estimator becomes a universal estimator of the directed information rate for certain predictors and classes of signals \cite{jiao2013universal}.

To assess the performance of an estimate of the causal measure, we define a notion of causality regret:
\begin{equation}
CR(n) \triangleq \sum_{i=1}^n \left| \estcyx(i) - \optcyx(i)  \right|
\end{equation}

\noindent where we define:
\begin{equation}
\optcyx(i) = \kl{\optfxc{i}}{\optfxr{i}}
\end{equation}

\noindent with $\optfxc{i} \in \refclass^{(c)}$ and $\optfxr{i} \in \refclass^{(r)}$ defined as the loss minimizing distributions from the complete and restricted reference classes. We note that with this notion of causal regret, the estimated causal measure is being compared against the best estimate of the causal measure from within a reference class. As such, we limit our consideration to the scenario in which the reference classes are sufficiently representative of the true sequences to produce a desirable $\optcyx$ (i.e. $\optcyx(i) \approx \cyx(i)$ for all $i$).

We now present the necessary preliminaries for proving a finite sample bound on the estimates of causality regret for the special case when $\mc{X}$ is a discrete space. We begin by introducing two assumptions.

\begin{assumption} \label{assumption:gbound}
For sequential predictors \estfxc{i} and \estfxr{i} and observations $(x^n,y^n,z^n)\in \mc{X}^n \times \mc{Y}^n \times \mc{Z}^n$, we assume that there exists some $L\in \mathbb{R}$ for which the collection of observations is such that:
\begin{equation}
\sup_{x \in \mc{X}} \left| \log \frac{\estfxc{i}(x)}{\estfxr{i}(x)} \right| < L \ \ \forall i=1,\dots,n
\end{equation}
\end{assumption}

\noindent Noting that $L$ is lower bounded by $\estcyx$, Assumption \ref{assumption:gbound} implies (given its role in Theorem \ref{thm:main_result}) that larger levels of causal influence take longer to estimate accurately.

\begin{assumption} \label{assumption:referencekl}
For loss minimizing distributions $\optfxc{i} \in \refclass^{(c)}$ and $\optfxr{i} \in \refclass^{(r)}$, restricted sequential predictor \estfxr{i}, and observations $(x^n,y^n,z^n)\in \mc{X}^n \times \mc{Y}^n \times \mc{Z}^n$:
\begin{equation}
\sum_{i=1}^n\left|E_{\optfxc{i}}\left[
r(\estfxr{i},\optfxr{i},X)
\right]\right| \le M^{(r)}(n)
\end{equation}
\end{assumption}

\noindent While it is understood that the expected regret is in general bounded by worst case regret, assumption \ref{assumption:referencekl} requires that the reference classes are sufficiently rich that the expected regret is not too large in \emph{absolute value}. This is necessary in bounding the causality regret because unlike the regret defined by \eqref{optimal_f}, $CR(n)$ \emph{increases} when the estimated distributions outperform the regret minimizing distributions.

We now show that the cumulative KL divergence from the best reference distribution to the predicted distribution is less than the predictor's worst-case regret.
\begin{lemma}\label{lemma:kl}
For a sequential predictor $\hat{f}_i$ with worst case regret $M(n)$, a collection observations $(x^n,y^n,z^n)$, and any distribution from the reference class $f \in \refclass_n$:
\begin{equation}
\sum_{i=1}^n \kl{f_i}{\hat{f}_i}\le M(n)
\end{equation}
\end{lemma}

\begin{proof}
\begin{align*}
\sum_{i=1}^n \kl{f_i}{\hat{f}_i}
&= \sum_{i=1}^n \sum_{x\in\mc{X}} f_i(x)
    \log \frac{f_i(x)}{\hat{f}_i(x)} \\
&\le \sum_{i=1}^n
    \left[ \sup_{x\in\mc{X}}
    \log \frac{f_i(x)}{\hat{f}_i(x)} \right]
    \sum_{x\in\mc{X}} f_i(x) \\
&= \sum_{i=1}^n \sup_{x\in\mc{X}} r(\hat{f}_i,f_i,x) \\
&\le \sup_{x^n\in\mc{X}^n} \sum_{i=1}^n r(\hat{f}_i,f_i,x_i)\\
&\le \sup_{x^n\in\mc{X}^n} \sup_{f\in\refclass_n} \sum_{i=1}^n r(\hat{f}_i,f_i,x_i)\\
&\le M(n)
\end{align*} \end{proof}
Next, we bound the cumulative difference in expectation of a bounded function between the best reference distribution and sequential predictor.

\begin{lemma} \label{lemma:g_func}
For a sequential predictor $\hat{f}_i$ with worst case regret $M(n)\ge 1$, a collection observations $(x^n,y^n,z^n)$, cumulative loss minimizing distribution $f^*_i$, and bounded functions $g_i:\mc{X}\rightarrow [-K,K]$ with $K\in\mathbb{R}$:
\begin{equation}
\sum_{i=1}^n \left| E_{f^*_i}[g_i(X)] -
    E_{\hat{f}_i}[g_i(X)] \right| \le
    \frac{\left|\mc{X}\right|K}{\sqrt{2}}\sqrt{n \cdot M(n)}
\end{equation}
\end{lemma}

\begin{proof}
\begin{align}
\sum_{i=1}^n &\left| E_{f^*_i} [g_i(X)] -
    E_{\hat{f}_i}[g_i(X)] \right| \nonumber \\
&= \sum_{i=1}^n \left| \sum_{x\in\mc{X}}
    \left[ f^*_i(x) - \hat{f}_i(x) \right] g_i(x) \right| \nonumber \\
&\le \sum_{i=1}^n \sum_{x\in\mc{X}} \left| f^*_i(x) -
    \hat{f}_i(x) \right| \left| g_i(x) \right|
    \label{ref_triangle_eq}\\
&\le \sum_{i=1}^n \sum_{x\in\mc{X}}
    K\sqrt{\frac{1}{2}\kl{f^*_i}{\hat{f}_i}}
    \label{ref_pinsker_assumption}\\
&= \frac{\left|\mc{X}\right| K}{\sqrt{2}} \sum_{i=1}^n
    \sqrt{\kl{f^*_i}{\hat{f}_i}} \nonumber
\end{align}

\noindent where \eqref{ref_triangle_eq} uses the triangle inequality and \eqref{ref_pinsker_assumption} uses Pinsker's inequality and the boundedness of $g_i$. Focusing on the sum, we define $\vec{v}$ such that $\vec{v}_i = \sqrt{\kl{f^*_i}{\hat{f}_i}}$ for $i=1,\dots,n$:
\begin{align}
\sum_{i=1}^n \sqrt{\kl{f^*_i}{\hat{f}_i}}
&= \left|\left|\vec{v}\right|\right|_1 \nonumber \\
&\le \sqrt{n}\left|\left|\vec{v}\right|\right|_2 \label{holders}\\
&= \sqrt{n}\left( \sum_{i=1}^n \kl{f^*_i}{\hat{f}_i}
    \right)^{\frac{1}{2}} \nonumber \\
&\le \sqrt{n \cdot M(n)} \label{ref_kl_lemma}
\end{align}

\noindent where \eqref{holders} uses H\"{o}lders inequality and \eqref{ref_kl_lemma} uses Lemma \ref{lemma:kl} and the assumption that $M(n) \ge 1$. \end{proof}

Finally, we can utilize the assumption and lemmas to bound the cumulative causality regret:

\begin{theorem} \label{thm:main_result}
Let the worst case regret for the predictors $\estfxr{i}$ and $\estfxc{i}$ be bounded by $R^{(r)}_n(\refclassr_n) \le M^{(r)}(n)$ and $R^{(c)}_n(\refclassc_n) \le M^{(c)}(n)$, respectively. Then, for any collection of observations $(x^n,y^n,z^n)\in \mc{X}^n \times \mc{Y}^n \times \mc{Z}^n$ such that Assumption \ref{assumption:gbound} holds with bound $L$, we have:
\begin{equation} \label{causal_bound}
\begin{aligned}
\sum_{i=1}^n & \left| \estcyx(i) - \optcyx(i) \right| \le \\
&M^{(c)}(n) + M^{(r)}(n) + 
    \frac{\left|\mc{X}\right|L}{\sqrt{2}}\sqrt{n \cdot M^{(c)}(n)}
\end{aligned}
\end{equation}
\end{theorem}

\begin{proof}
We begin by defining the functions:
\begin{equation*}
\hat{g}_i(X) \triangleq \log \frac{\estfxc{i}(X)}{\estfxr{i}(X)} \ \ \ \ \
g^*_i(X) \triangleq \log \frac{\optfxc{i}(X)}{\optfxr{i}(X)}.
\end{equation*}

\noindent Using the definition of the causal measure and KL-divergence:
\begin{align}
&\sum_{i=1}^n \left| \estcyx(i) - \optcyx(i) \right|
\nonumber \\
& \ \ \ \ \ \ \ \ \ \ - \left|
E_{\optfxc{i}}\left[ \hat{g}_i(X)\right] -
E_{\estfxc{i}}\left[ \hat{g}_i(X)\right]
\right| \label{beginning} \\
&= \sum_{i=1}^n
\left|
E_{\optfxc{i}}\left[ g^*_i(X)\right] -
E_{\estfxc{i}}\left[ \hat{g}_i(X)\right]
\right| \nonumber \\
& \ \ \ \ \ \ \ \ \ \ - \left|
E_{\optfxc{i}}\left[ \hat{g}_i(X)\right] -
E_{\estfxc{i}}\left[ \hat{g}_i(X)\right]
\right| \nonumber \\
&\le \sum_{i=1}^n
\bigg| \left|
E_{\optfxc{i}}\left[ g^*_i(X)\right] -
E_{\estfxc{i}}\left[ \hat{g}_i(X)\right]
\right| \nonumber \\
& \ \ \ \ \ \ \ \ \ \ - \left|
E_{\optfxc{i}}\left[ \hat{g}_i(X)\right] -
E_{\estfxc{i}}\left[ \hat{g}_i(X)\right]
\right| \bigg| \label{absval} \\
&\le \sum_{i=1}^n
\bigg|
E_{\optfxc{i}}\left[ g^*_i(X)\right] -
E_{\estfxc{i}}\left[ \hat{g}_i(X)\right] \nonumber \\
& \ \ \ \ \ \ \ \ \ \ -
E_{\optfxc{i}}\left[ \hat{g}_i(X)\right] +
E_{\estfxc{i}}\left[ \hat{g}_i(X)\right] \bigg| \label{rev_tri} \\
&= \sum_{i=1}^n \left|
E_{\optfxc{i}}\left[ g^*_i(X) - \hat{g}_i(X)\right] \right| \nonumber \\
&= \sum_{i=1}^n \left|
E_{\optfxc{i}}\left[
\log \frac{\optfxc{i}(X)}{\estfxc{i}(X)}
- \log \frac{\optfxr{i}(X)}{\estfxr{i}(X)}
\right] \right| \nonumber \\
&\le \sum_{i=1}^n \left|
\kl{\optfxc{i}}{\estfxc{i}} \right|+
\left|
E_{\optfxc{i}}\left[
\log \frac{\optfxr{i}(X)}{\estfxr{i}(X)}
\right]
\right| \label{tri} \\
&\le M^{(c)}(n) + M^{(r)}(n) \label{regretlemma}
\end{align}

\noindent where \eqref{absval} follows from the properties of absolute value, \eqref{rev_tri} follows from the reverse triangle inequality, \eqref{tri} follows from the triangle inequality, and \eqref{regretlemma} follows from non-negativity of the KL-divergence, Lemma \ref{lemma:kl}, and Assumption \ref{assumption:referencekl}. Moving the second term of \eqref{beginning} to the other side of the inequality yields:
\begin{align}
&\sum_{i=1}^n \left| \estcyx(i) - \optcyx(i) \right| \nonumber \\
& \ \ \le M^{(c)}(n) + M^{(r)}(n) + \left|
E_{\optfxc{i}}\left[ \hat{g}_i(X)\right] -
E_{\estfxc{i}}\left[ \hat{g}_i(X)\right]
\right| \nonumber \\
& \ \ \le M^{(c)}(n) + M^{(r)}(n) +
\frac{\left|\mc{X}\right|L}{\sqrt{2}}\sqrt{n \cdot M^{(c)}(n)}
\label{final_eq}
\end{align}

\noindent where \eqref{final_eq} follows from Assumption \ref{assumption:gbound} (boundedness of $\hat{g}_i$) and Lemma \ref{lemma:g_func}. This concludes the proof. \end{proof}

\section{Simulations} \label{simulations}

\begin{comment}
- Application to conditional bernoulli model using CTW algorithm (Fig 1)
- Explicit derivation of bounds and computation of regret
- Use changing environment meta algorithm for changing parameters to demonstrate that it appears to work even though theorem doesn't apply (Fig 2)
\end{comment}

\begin{figure*}
\includegraphics[width=\linewidth]{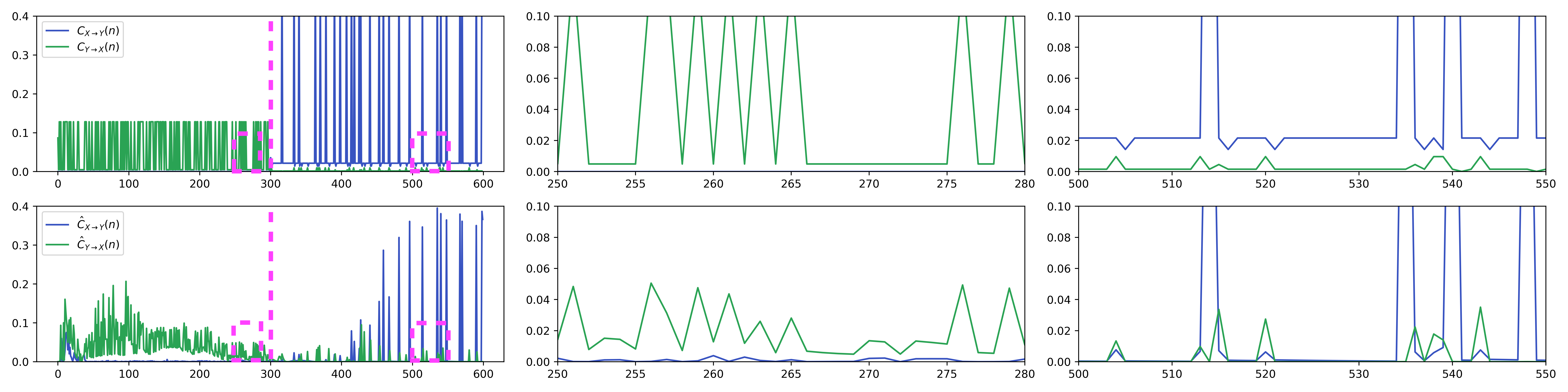}
\caption{True (top) and estimated (bottom) causal measures for entire time series (left) and selections pre (center) and post (right) parameter change point. Empirical values of $L$ (Assumption \ref{assumption:gbound}) are 1.763 and 1.227 for $X\rightarrow Y$ and $Y\rightarrow X$, respectively.}
\label{fig:results}
\end{figure*}

We begin by demonstrating the estimation of the proposed causal measure on a pair of jointly Markov binary processes $X^n$ and $Y^n$ that undergo a change point in the underlying parameters. For $j\in \{1,2\}$, we use a logistic model to represent the probabilities that $X_i$ and $Y_i$ are equal to one given the complete history:
\begin{align*}
p_{X_i}^{(c)}(x_{i-1},y_{i-1}) \triangleq
\frac{e^{\theta^j_x+\theta^j_{xx}x_{i-1}+\theta^j_{yx}y_{i-1}}}
{1+e^{\theta^j_x + \theta^j_{xx}x_{i-1}+ \theta^j_{yx}y_{i-1}}} \\
p_{Y_i}^{(c)}(x_{i-1},y_{i-1}) \triangleq
\frac{e^{\theta^j_y + \theta^j_{yy}y_{i-1}+ \theta^j_{xy}x_{i-1}}}
{1+e^{\theta^j_y + \theta^j_{yy}y_{i-1}+ \theta^j_{xy}x_{i-1}}}
\end{align*}

\noindent To compute the true causal measure, we additionally need the restricted probabilities. It is important to note that the joint-Markovicity does not imply that the processes are individually Markov. As such, the restricted probability that $X_i$ is equal to one given the restricted history is defined using a recursively updated distribution over the hidden $Y_i$:
\begin{eqnarray*}
p_{X_i}^{(r)}(x^{i-1}) \triangleq
p^{(h)}_{Y_{i-1}}\cdot p_{X_i}^{(c)}(x_{i-1},1) +
\bar{p}^{(h)}_{Y_{i-1}}\cdot p_{X_i}^{(c)}(x_{i-1},0)
\\
p^{(h)}_{Y_i} \triangleq
p^{(h)}_{Y_{i-1}}\cdot p_{Y_i}^{(c)}(x_{i-1},1) +
\bar{p}^{(h)}_{Y_{i-1}}\cdot p_{Y_i}^{(c)}(x_{i-1},0)
\end{eqnarray*}

\noindent where $p^{(h)}_{Y_i}$ is the probability of the \emph{hidden} $Y_i$ being one and $\bar{p}^{(h)}_{Y_i} \triangleq 1-p^{(h)}_{Y_i}$. We can define $p_{Y_i}^{(r)}(y^{i-1})$ similarly using $p^{(h)}_{X_i}$.

To estimate the causal measure (in both directions) we estimate both the complete and restricted probabilities using a Bayesian updating scheme over a discretized parameter space with a uniform prior. To accommodate the parameter change point, we incorporate the shrinking to the prior technique \cite{sancetta2012universality} with $\alpha=0$ and $\lambda=0.9999$ to the updating procedure. For space considerations we have omitted further details of experiments and provide detailed code on Github\footnote{\url{https://github.com/gabeschamberg/sample-path-causal-measure/blob/master/notebooks/isit_simulation.ipynb}}.

Figure \ref{fig:results} shows the true and estimated causal measures. We can see that the spikes in causal influence are captured by the estimate. This example illustrates that the proposed causal measure is not immune to the difficulties of change point scenarios in that it takes roughly 100 samples after the change point for the estimator to adapt. However, a key point is that once the estimator does adapt, the temporal resolution is much better than what could be expected from windowing techniques, as the infrequent spikes seen in rightmost column of Figure \ref{fig:results} are localized to a single time point.

\section{Discussion} \label{discussion}

We have presented a non-negative measure of local causal influence that captures the time-varying nature of causal relationships that is inherent to both stationary and non-stationary settings. Furthermore, we have shown that under mild assumptions, the finite sample performance of an estimator of the measure can be determined as a function of the regret of the sequential predictors used to implement the estimator.

It is important to note that the proposed causal measure does not solve the problem of estimating causal influence in time-varying settings, but rather it provides a perspective on causal influence that is naturally extended to any setting for which there are effective sequential prediction algorithms. By conditioning on the observed past, we avoid the need to decide a window length (to approximate an expectation) when estimating DI and TE in a time-varying setting.

Further research includes calculating the causal regret for specific estimators and carefully characterizing the circumstances for which the assumptions hold. Additionally, estimation of the measure on real data would enable moving past simply identifying the direction of information flow between real-word processes to identifying particular patterns for which the causal influence is greatest.

%%%%%%%%%%%%%%%%%%%%%%%%%%%%%%%%%%%%%%%%%%%%%%%%%%%%%%%%%%%%%%%%%%%%%%%%%%%%%%%

\printbibliography

\end{document}